\documentclass[journal]{IEEEtran}
 \pdfoutput=1
\usepackage{cite}
\usepackage[numbers,sort&compress]{natbib}

\usepackage{ifpdf}
\ifCLASSINFOpdf
\usepackage[pdftex]{graphicx}
\graphicspath{{../pdf/}{../jpeg/}}
\DeclareGraphicsExtensions{.pdf,.jpeg,.png}
\else
\usepackage[dvips]{graphicx}
\graphicspath{{../eps/}}
\DeclareGraphicsExtensions{.eps}
\fi
\usepackage{epstopdf}

\usepackage{amsthm}
\usepackage{amsmath}
\usepackage{amssymb}
\usepackage{amsfonts}

\usepackage{algorithm}
\usepackage{algorithmic}

\usepackage{tabularx}
\usepackage{booktabs}
\usepackage{array}
\usepackage{multirow}
\usepackage{makecell}

\usepackage{float}
\usepackage[caption=false,font=normalsize,labelfont=sf,textfont=sf]{subfig}

\usepackage{CJKutf8}

\usepackage{xcolor}
\usepackage{soul}

\usepackage{url}
\usepackage{textcomp}
\usepackage{verbatim}

\theoremstyle{plain}
\newtheorem{proposition}{Proposition}

\def\BibTeX{{\rm B\kern-.05em{\sc i\kern-.025em b}\kern-.08em
		T\kern-.1667em\lower.7ex\hbox{E}\kern-.125emX}}

\begin{document}
\renewcommand{\qedsymbol}{}
\title{ Low-PAPR OFDM-ISAC Waveform Design Based on Frequency-Domain Phase Differences}

\author{Kaimin~Li,
	Jiahuan~Wang,
	Haixia~Cui, Bingpeng Zhou and Pingzhi Fan
	\thanks{Kaimin Li, Jiahuan Wang and Haixia Cui are with the School of Electronic Science and Engineering (School of Microelectronics), South China Normal University, Foshan 528225, China (Email: 2023025040@m.scnu.edu.cn, jiahuanwang@m.scnu.edu.cn, cuihaixia@scnu.edu.cn).}
	\thanks{Bingpeng Zhou is with the School of Electronics and Communication Engineering, Sun Yat-sen University, Shenzhen 518000, China (email:zhoubp3@mail.sysu.edu.cn).}
	\thanks{Pingzhi Fan is with the School of Information Science and Technology, Southwest Jiaotong University, Chengdu 611756, China (email: pzfan@swjtu.edu.cn).}
}

\markboth{Journal of \LaTeX\ Class Files,~Vol.~xx, No.~xx, Dec~2024}%
{Shell \MakeLowercase{\textit{et al.}}: Bare Demo of IEEEtran.cls for IEEE Journals}

\maketitle
\begin{abstract}
	Low peak-to-average power ratio (PAPR) orthogonal frequency division multiplexing (OFDM) waveform design is a crucial issue in integrated sensing and communications (ISAC). This paper introduces an OFDM-ISAC waveform design that utilizes the entire spectrum simultaneously for both communication and sensing by leveraging a novel degree of freedom (DoF): the frequency-domain phase difference (PD). Based on this concept, we develop a novel PD-based OFDM-ISAC waveform structure and utilize it to design a PD-based Low-PAPR OFDM-ISAC (PLPOI) waveform. The design is formulated as an optimization problem incorporating four key constraints: the time-frequency relationship equation, frequency-domain unimodular constraints, PD constraints, and time-domain low PAPR requirements. To solve this challenging non-convex problem, we develop an efficient algorithm,  ADMM-PLPOI, based on the alternating direction method of multipliers (ADMM) framework. Extensive simulation results demonstrate that the proposed PLPOI waveform achieves significant improvements in both PAPR and bit error rate (BER) performance compared to conventional OFDM-ISAC waveforms.
\end{abstract}


%
\IEEEpeerreviewmaketitle

\section{Introduction}
%
%
%
%


\IEEEPARstart{I}{ntegrated} {sensing and communication (ISAC) has emerged as a key technology for next-generation wireless networks, driven by the growing demand for spectrum resources \cite{liu2022integrated,akan2020internet,feng2020joint,zhang2021overview,liu2022survey,liu2020joint}, which improves system efficiency and meets the practical requirements of sixth-generation (6G) applications such as vehicle-to-everything \cite{meng2024sensing} and massive Internet of Things (IoT) \cite{qi2020integrated}. The core of ISAC implementation is waveform design  \cite{zhou2022integrated}, which must strike a balance between high-rate data transmission and accurate target detection while addressing practical challenges such as interference and hardware limitations. The OFDM waveform, employed in existing 4G and 5G systems, is a leading candidate for 6G ISAC due to its multi-carrier structure, which provides robustness against frequency-selective fading and supports flexible subcarrier allocation for efficient resource sharing \cite{chen2024}. However, the high peak-to-average power ratio (PAPR) of OFDM signals poses a challenge for ISAC waveform design, as it reduces the power amplifier efficiency and induces nonlinear distortion, which degrades the communication and sensing performance in ISAC systems \cite{varshney2023}.}
\begin{figure*}[htbp]
	\centering
	\includegraphics[width=1\textwidth]{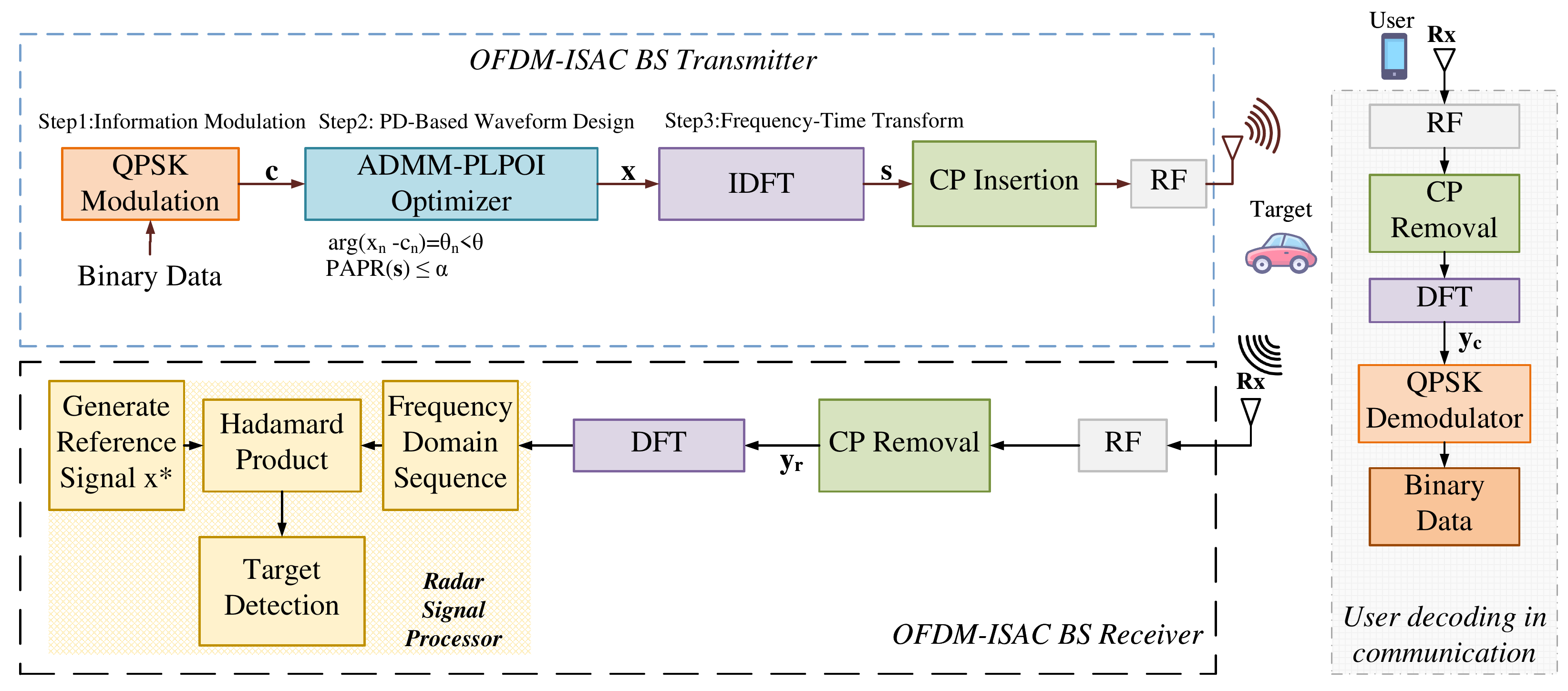}
	\caption{System model.}
	\label{fig:sysmod0}
\end{figure*}

Reducing PAPR has always been a critical issue in OFDM waveform design  \cite{hao2021,siluveru2024,tank2024,cinemre2024,arulkarthick2024}. In conventional communication systems, methods for reducing PAPR primarily include three approaches: signal distortion techniques \cite{siluveru2024}, multiple signaling and coding techniques  \cite{cinemre2024}, and probability techniques  \cite{tank2024}. For signal distortion techniques, recent research focuses on adaptive gain adjustment methods \cite{siluveru2024} and the combination of windowing functions with clipping parameters \cite{arulkarthick2024}, which effectively reduce PAPR while maintaining acceptable BER performance. In coding techniques, advances have been made through Gaussian matrix-based precoding \cite{cinemre2024} and hybrid precoded-companding schemes  \cite{arulkarthick2024}, demonstrating improved PAPR reduction with enhanced BER performance. For multiple signaling techniques, modified selective mapping approaches \cite{tank2024} have been developed to reduce computational complexity while maintaining PAPR reduction capability, and lexicographical permutation methods \cite{tank2024} have shown promise in reducing system complexity by up to 90$\%$ while achieving comparable error performance to traditional selective mapping.

However, PAPR reduction techniques used in communication-only OFDM systems cannot be directly applied to OFDM-ISAC systems, as OFDM-ISAC systems must simultaneously satisfy communication performance and sensing performance. To meet these dual-function requirements, several waveform design methods specifically for OFDM-ISAC systems have been proposed recently. In \cite{huang2022}, the authors aimed to achieve low PAPR while maintaining sensing capability. They considered a flexible RadCom structure in which the communication subcarriers are located in continuous radar frequency bands and proposed an $l$-norm cyclic algorithm (LNCA) to design waveforms that meet the low PAPR requirements. In \cite{varshney2023}, the work focused on enhancing PAPR optimization under zero correlation sidelobe level constraint by introducing a direct optimization approach based on MM technique, which enables PAPR values to approach 1.  Furthermore, to address computational complexity in large-scale scenarios, the authors in \cite{wu2024} proposed an ADMM-based method that significantly reduces complexity while maintaining fast convergence with large numbers of subcarriers. In addition, in \cite{bazzi2023}, an ADMM-based DFRC waveform design method was proposed that minimizes multi-user interference under radar similarity and PAPR constraints, and an iterative algorithm was developed to obtain a stable waveform with desired PAPR characteristics. In \cite{hu2022}, a joint waveform optimization framework was proposed for MIMO-OFDM based DFRC systems to minimize PAPR while maintaining dual-functional performance. To solve this non-convex problem, an SDR-based algorithm was proposed to decompose the original problem into parallel sub-problems, which can effectively achieve low PAPR.


Nevertheless, existing research mainly focuses on various approaches for PAPR reduction in ISAC systems. Some methods partition the spectrum into a communication band and an optimization band, where the PAPR reduction performance is influenced by the ratio between these bands. Another category of approaches includes a weighted optimization framework \cite{liu2018} that balances radar and communication performance using a weighting parameter, and an extended formulation with PAPR constraints for OFDM signals \cite{chen2024}. In practical applications, if the parameters are not properly selected, achieving satisfactory PAPR optimization becomes increasingly challenging.

To address these limitations, this paper proposes a unified OFDM-ISAC waveform design method based on frequency-domain phase difference (PD). Unlike existing approaches, our method enables simultaneous utilization of the entire frequency band for both communication and sensing purposes without weighting parameter. By introducing the concept of frequency-domain PD between the transmit signal and preset information signal, we develop a novel optimization framework that jointly considers PAPR reduction, spectral efficiency, and sensing performance. Through meticulous selection of the PD threshold and optimization of waveform parameters, our approach effectively realizes integrated sensing and communication functionality while maintaining system performance. The main contributions can be summarized as follows:

\begin{itemize}
	\item We propose a novel low PAPR OFDM-ISAC waveform design that leverages frequency-domain PD as a new degree of freedom (DoF). Unlike traditional approaches that rely on spectrum partitioning, our approach utilizes the entire spectrum for both communication and sensing simultaneously, based on PD, allowing for more efficient spectrum usage.

	\item To address the non-convex optimization problem with coupled constraints, we develop an ADMM-based algorithm that solves the parameter coupling between time-domain PAPR and frequency-domain PD constraints through variable splitting, achieving efficient convergence with FFT-based implementation.
	
	\item The proposed PD-based low-PAPR OFDM-ISAC waveform design effectively reduces PAPR without compromising spectral efficiency. By adjusting the PD threshold $\theta$ within its feasible range, we can control PAPR while maintaining the constant modulus property in the frequency domain, which is critical for sensing.
	
	\item The proposed waveform demonstrates excellent dual-functional performance with good BER performance in communication and desirable periodic auto-correlation characteristics in sensing, validating its effectiveness for practical OFDM-ISAC systems.
\end{itemize}

The rest of this paper is organized as follows: Section II introduces the OFDM-ISAC system and waveform structure. In Section III, we propose the OFDM-ISAC  waveform optimization problem and introduce the ADMM framework to solve it. In Section IV, Simulation verification are presented to show the superiority of the proposed OFDM-ISAC waveform. Finally, the paper's conclusions are detailed in Section VI.

\emph{Notations}: Throughout this manuscript, bold lowercase letters represent vectors, while bold uppercase letters represent matrices The symbols $(\cdot)^*$, $(\cdot)^T$ and $(\cdot)^H$ denote the conjugate, transpose and conjugate transpose, respectively. The 2-norm of a vector $\mathbf{a}$ and the $\infty$ -norm of a vector $\mathbf{a}$ are indicated by $\|\mathbf{a}\|_2$ and $\|\mathbf{a}\|_\infty$, respectively. The operation $\circ$ denotes the Hadamard product.

\section{OFDM ISAC Systems And Proposed Waveform Structure}
\subsection{OFDM-ISAC System Model}
\begin{figure*}[htbp]
	\centering
	\includegraphics[width=1\textwidth]{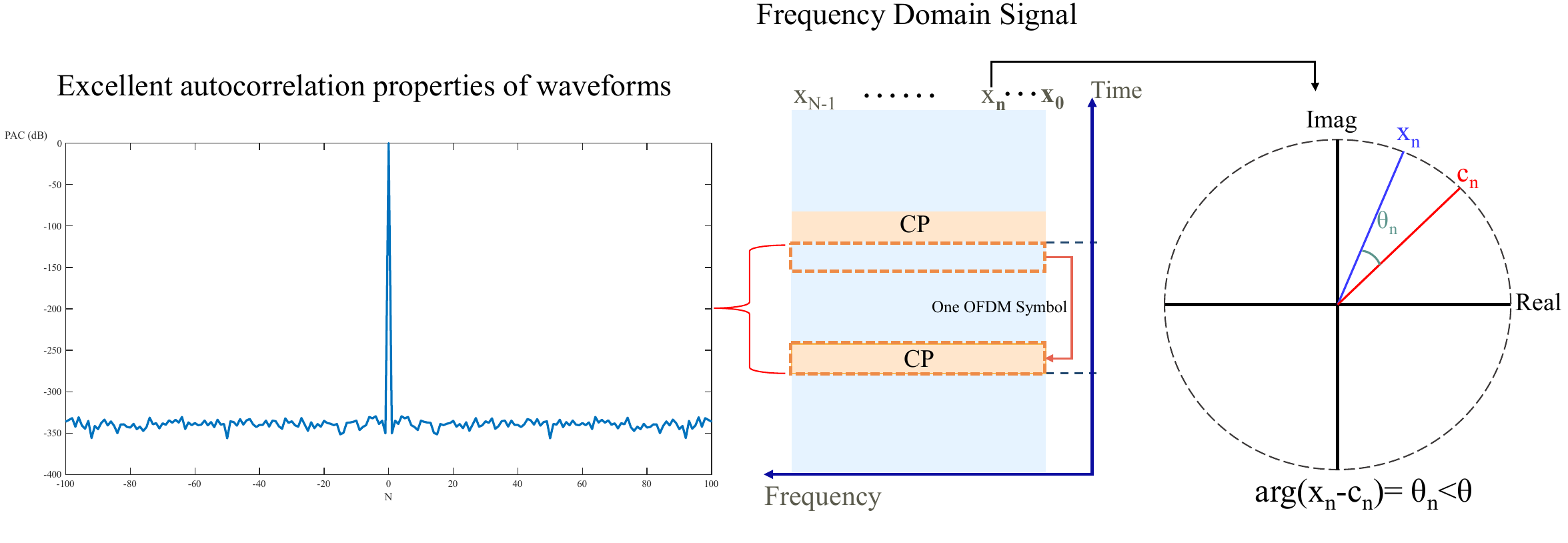}
	\caption{OFDM-ISAC signal structure.}
	\label{xicp}
\end{figure*}
Consider an OFDM-ISAC system, which consists of a base station (BS), a communication user, and multiple targets,  as illustrated in Fig. \ref{fig:sysmod0}. The BS functions as both a  communication transmitter and a sensing receiver. Specifically, it is equipped with a transmit antenna to send OFDM-ISAC signals and a receive antenna to capture echoes \cite{huang2022}. 


In the transmitter, the binary data is modulated via QPSK to obtain the intended frequency-domain communication sequence $\mathbf{c}=[c_0,c_1,\cdots, c_{N-1}]^T$, where $N$ is the number of subcarriers. Then the proposed PD-based waveform design scheme\footnote{The details of the PD-based waveform design scheme are provided in Sections II-B and III} is applied to $\mathbf{c}$, resulting in the desired frequency-domain ISAC signal $\mathbf{x}\in \mathbb{C}^{N}$. After performing an $M$-points inverse discrete Fourier transform (IDFT), the OFDM symbol $\mathbf{s}\in \mathbb{C}^{M}$ is obtained.  The symbol $\mathbf{s}$ is then transmitted after cyclic prefix (CP) insertion via the radio frequency (RF) chain, where CP is denoted as $\mathbf{s}_{CP}\in \mathbb{C}^{N_{CP}}$ and $N_{CP}$ is the length of CP.

In the communication receiver,  after CP removal operations, the discrete-time received signal undergoes discrete Fourier transform (DFT) transformation to yield its frequency-domain representation \cite{li1999channel}:
\begin{equation}
	\mathbf{y}_c = h_c\mathbf{x}+\mathbf{w}_c,\label{eq:yc}
\end{equation}
where $h_c$ is the communication channel frequency response and $\mathbf{w}_c$ denotes the frequency-domain noise.  
In the sensing receiver, the time-domain  echoes reflected by the targets can be represented by $\mathbf{y}_r$. The sensing receiver processes the reflected signal by extracting a segment starting $N_{CP}$ after transmission with a length of $N$ \cite{huang2022}. The extracted signal is then transformed into the frequency domain using a DFT, where it undergoes matched filtering by multiplying with the conjugate of $\mathbf{x}$. Subsequently, an IDFT is applied to the product to convert it back to the time domain, generating a range profile in which peaks correspond to target delays and distances. This approach effectively implements pulse compression, thereby enhancing resolution and target detectability.  The Doppler shift is estimated by leveraging the DFT of the signal samples on each subcarrier \cite{huang2022}. 

\subsection{Proposed PD-based OFDM-ISAC Waveform Structure}

In order to fully utilize spectrum resources, we propose a new OFDM-ISAC waveform structure that simultaneously meets sensing and communication requirements.  The proposed waveform structure in frequency domain is given by
\begin{equation}
	\mathbf{x} = \mathbf{c}\circ \mathbf{e}, \label{eq: xce}
\end{equation}
where $\mathbf{x}$ is required to be unimodular, i.e.,
\begin{equation}
	|x_n|=1, n=0,1,2,\cdots, N-1, \label{eq:xuni}
\end{equation}
which ensures desirable time-domain periodic auto-correlation (PAC) properties to enhance the sensing performance \cite{huang2022}. Besides, $\mathbf{e} = [e^{j\theta_0},e^{j\theta_1},\cdots,e^{j\theta_{N-1}}]^T$ represents the phase shift sequence and  $\theta_n$ denotes the phase difference between $c_n$ and $x_n$, i.e.,
  \begin{equation}
  	\mathrm{arg}(x_n-c_n)=\theta_n, \,\, n=0, 1, \cdots, N-1.
  \end{equation}

By substituting \eqref{eq: xce} into \eqref{eq:yc}, the frequency-domain received signal at the communication receiver is given by
\begin{equation}
	{y}_n^c = h_c c_n e^{jn\theta_n} +w_n^c, n=0.1,\cdots, N-1, \label{eq:ync}
\end{equation}
where ${y}_n^c$ is the $n$-th element of $\mathbf{y}_c$ and $w_n^c$ is the $n$-th element of $\mathbf{w}_c$.  In order to correctly demodulate the symbol $c_n$  from ${y}_n^c$,  the phase difference $\theta_n$ should be constrained by a threshold $\theta \in (0,\pi/4)$, i.e.,
\begin{equation}
	\mathrm{arg}(x_n-c_n)=\theta_n<\theta. \label{ineq:arg}
\end{equation}


As illustrated in  Fig.~\ref{xicp}, the core idea of the proposed method is to design $\mathbf{x}$ based on the communication symbols $\mathbf{c}$, such that 
$\mathbf{x}$ possesses both communication and sensing capabilities. This is achieved by applying a phase diffenece $\theta_n$ to each communication symbol $c_n$. In other words, flexible design of $\mathbf{x}$ can be achieved by adjusting $\theta_n$ within an appropriate threshold $\theta$, which provides degrees of freedom (DoF) for the design of $\mathbf{x}$. Furthermore, by fully exploiting this DoF, the proposed waveform can effectively control the PAPR in the time domain,  as discussed below.



For a given $\mathbf{x}\in \mathbb{C}^{N}$, the PAPR of an OFDM waveform is equivalent to that of the four-times oversampled discrete signal $\mathbf{s}\in \mathbb{C}^{M}$, where $M=4N$ \cite{huang2022},  $\mathbf{s}$ is the IDFT of $\mathbf{x}$, i.e.,
\begin{equation}
	\mathbf{s} = \mathbf{Ax},\label{eq:sAx}
\end{equation}
and $\mathbf{A}\in \mathbb{C}^{M\times N}$ represents the IDFT matrix \cite{wang2019}:
\begin{equation}\mathbf{A}=\begin{bmatrix}1&1&1&\cdots&1\\1&e^{j2\pi/M}&e^{j2\pi/2/M}&\cdots&e^{j2\pi(N-1)/M}\\1&e^{j2\pi2/M}&e^{j2\pi4/M}&\cdots&e^{j2\pi2(N-1)/M}\\\vdots&\vdots&\vdots&\cdots&\vdots\\1&e^{\frac{j2\pi(M-1)}M}&e^{\frac{j2\pi2(M-1)}M}&\cdots&e^{\frac{j2\pi(M-1)(N-1)}M}\end{bmatrix}.\end{equation}
The PAPR of $\mathbf{s}$ is defined as \cite{tian2021}, \cite{wang2021}
\begin{equation}
	\mathrm{PAPR}(\mathbf{s})=\frac{\max_{m=0,1,...,M-1}|s_m|^2}{\frac1{M}\sum_{m=0}^{M-1}|s_m|^2}=\frac{\|\mathbf{s}\|_\infty^2}{\frac1{M}\|\mathbf{s}\|_2^2}.
\end{equation}
The low-PAPR requirement for $\mathbf{s}$ can be expressed as 
\begin{equation}
	\mathrm{PAPR}(\mathbf{s})\leq \alpha,\label{ineq:papr}
\end{equation}
where $\alpha$ represents the PAPR threshold.

\subsection{Problem Formulation}
Based on the proposed PD-based OFDM-ISAC waveform structure in Section II-B,  we formulate the following optimization problem to find the PD-based low-PAPR OFDM-ISAC (PLPOI)  waveform $\mathbf{s}$ and its frequency-domain counterpart $\mathbf{x}$:
\begin{subequations}\label{opt0}
	\begin{align}
		\boldsymbol{P}_0:\,\,	\textbf{find}\quad&\textbf{x , s}\\
		\mbox{s.t.}\quad\
		&\mathbf{A}\mathbf{x}=\mathbf{s}\label{opt0:b}\\
		&\mathrm{PAPR}(\mathbf{s})\leq\alpha \label{opt0:c}\\
		&\arg(x_{n}-c_{n})<\theta,\,\, n=0,1,\cdots,N-1\label{opt0:d}\\
		&|x_n|=1,\,\, n=0,1,\cdots,N-1.\label{opt0:e}
	\end{align}
\end{subequations}
where \eqref{opt0:b} represents  the time-frequency relationship equation. The low PAPR constraint and PD constraint are guaranteed by \eqref{opt0:c} and \eqref{opt0:d}, respectively. Furthermore, the unimodular constraint \eqref{opt0:e} is imposed to  guarantee  perfect PAC properties.

The optimization problem $\boldsymbol{P}_0$ is challenging due to the fact that i) frequency-domain signal \(\mathbf{x}\) and time-domian signal \(\mathbf{s}\) are tightly coupled in the constraint \eqref{opt0:b}; ii) the fractional form of PAPR for time-domain signal $\mathbf{s}$ makes the PAPR inequality constraint \eqref{opt0:c} non-convex;
		iii) the PD inequality and  unimodular constraint on frequency-domain signal $\mathbf{x}$ make the \eqref{opt0:d} and \eqref{opt0:e} both non-convex.

\section{Design of PD-based Low-PAPR OFDM-ISAC Waveforms Using an ADMM-Based Approach}
In this section, we address the design of a frequency-domain PD-based low-PAPR OFDM-ISAC (PLPOI) waveform based on the proposed ADMM-based algorithm, referred to as ADMM-PLPOI.

\subsection{The Proposed ADMM-PLPOI Algorithm}
To address these challenges effectively, 
we propose an efficient iterative algorithm based on ADMM framework. Firstly, the ADMM framework is employed to decouple the variables $\mathbf{x} $ and $\mathbf{s}$, which can transform the original problem into two subproblems that are easier to solve independently \cite{boyd2011}.  In the following, we introduce the ADMM framework:
\begin{subequations}\label{Zd}
	\begin{align}
		&\mathbf{x}^{(k+1)}=\underset{\mathbf{x}\in\mathcal{X}}{\operatorname*{argmin}}L_{\rho}\left(\mathbf{x},\mathbf{s}^{(k)},\mathbf{y}^{(k)}\right), \label{Za}\\
		&\mathbf{s}^{(k+1)}=\underset{\mathbf{s}\in\mathcal{S}}{\operatorname*{\mathrm{argmin}}}L_{\mathbf{\rho}}\left(\mathbf{x}^{(k+1)},\mathbf{s},\mathbf{y}^{(k)}\right), \label{Zb} \\
		&\mathbf{y}^{(k+1)}=\mathbf{y}^{(k)}+\rho(\mathbf{A}\mathbf{x}^{(k+1)}-\mathbf{s}^{(k+1)}), \label{Zc}
	\end{align}
\end{subequations}
where $\mathbf{x} \in \mathcal{X}$ represents frequency-domain related constraints \eqref{opt0:d} and \eqref{opt0:e}, and $\mathbf{s} \in \mathcal{S}$  corresponds to time-domain related constraint \eqref{opt0:c}. In addition,
$L_{\rho}(\mathbf{x}, \mathbf{s}, \mathbf{y})$ in \eqref{Zd} is the augmented Lagrangian function, which can be expressed as \begin{equation}L_{\rho}(\mathbf{x},\mathbf{s},\mathbf{y})=\mathrm{Re}\{\mathbf{y}^{H}(\mathbf{A}\mathbf{x}-\mathbf{s})\}+\frac{\rho}{2}\|\mathbf{A}\mathbf{x}-\mathbf{s}\|_{2}^{2}. \label{eq:la}
\end{equation}
Here, the penalty parameter $\rho > 0$, $\mathbf{y}$ denotes the Lagrange multiplier, and $k$ indicates the iteration number. Subsequently, we only need to focus on solving subproblems \eqref{Za} and \eqref{Zb}.


\subsection{Solving Subproblem \eqref{Za}}
By substituting \eqref{eq:la} into \eqref{Za}, the subproblem \eqref{Za} can be transformed into
\begin{subequations}\label{Ba}
	\begin{align}
		\min_{\mathbf{x}\in\mathbf{C}^{N}}\quad&\mathrm{Re}\left\{\mathbf{y}^{(k)H}\left(\mathbf{A}\mathbf{x}-\mathbf{s}^{(k)}\right)\right\}+\frac{\rho}{2}\left\Vert\mathbf{A}\mathbf{x}-\mathbf{s}^{(k)}\right\Vert_{2}^{2}\label{Ca}\\
		\mbox{s.t.}\quad\
		&\arg(x_{n}-c_{n})<\theta,\,\, n=0,1,\cdots,N-1\label{BBb}\\
		&|x_n|=1,\,\, n=0,1,\cdots,N-1.\label{BBc}
	\end{align}
\end{subequations}
In fact, the problem \eqref{Ba} is still non-convex with respect to $\mathbf{x}$. In the following, to address this and derive the analytical solution, we first reformulate the objective function into a complete square form, allowing us to obtain a relaxed solution by extending the constraints to $\mathbb{C}^N$. Then, we project the relaxed solution onto the nonconvex feasible set defined by constraints \eqref{BBb} and \eqref{BBc}.

\emph{1) Objective Transformation and Relaxation} : For the objective transformation, we  present the following proposition:
\begin{proposition}
	The objective function in \eqref{Ca} can be  equivalent to the following form: 
	\begin{align}
		\frac{\rho M}{2} \left\|\mathbf{x}-M\mathbf{A}^{H}\left(\mathbf{s}^{(k)}-\frac{1}{\rho} \mathbf{u}^{(k)}\right)\right\|_{2}^{2}+\mathrm{const}.\label{opt:powerml}
	\end{align}%
	where $\mathrm{const}$ denotes a constant term.
\end{proposition}

\begin{proof}
	Based on the orthogonal property of IDFT matrix $\mathbf{A}$, we have
	\begin{align}
		\frac{1}{M}\mathbf{A}^{H}\mathbf{A}=\mathbf{I}_{N},\label{eq:oIDFT}
	\end{align}
	where $\mathbf{I}_{N}$ is the $N\times N$ identity matrix. Accordingly, we have
	\begin{equation}\begin{aligned}\label{ko}
			&\|\mathbf{Ax}-\mathbf{g}\|_{2}^{2} \\
			&=\mathbf{x}^H\mathbf{A}^H\mathbf{A}\mathbf{x}-2\mathrm{Re}\{\mathbf{x}^H\mathbf{A}^H\mathbf{g}\}+\|\mathbf{g}\|_2^2 \\
			&=M\|\mathbf{x}\|_2^2-2\mathrm{Re}\{\mathbf{x}^H\mathbf{A}^H\mathbf{g}\}+\|\mathbf{g}\|_2^2 \\
			&=M\|\mathbf{x}-\frac{1}{M}\mathbf{A}^H\mathbf{g}\|_2^2-\frac{1}{M}\|\mathbf{A}^H\mathbf{g}\|_2^2+\|\mathbf{g}\|_2^2,
	\end{aligned}\end{equation}
	where $\mathbf{g}\in \mathbb{C}^{M}$. 
	Then, based on \eqref{ko}, we complete the square for \eqref{Ca}:
	\begin{equation}
		\begin{split}
			&\text{Re}\{\mathbf{y}^{(k)H}(\mathbf{Ax}-\mathbf{s}^{(k)})\}+\frac{\rho}{2}||\mathbf{Ax}-\mathbf{s}^{(k)}||_{2}^{2} 
			\\&=\frac{\rho}{2}(2 \mathrm{Re}\{\frac{1}{\rho}\mathbf{y}^{(k)H}(\mathbf{Ax}-\mathbf{s}^{(k)})\}+||\mathbf{Ax}-\mathbf{s}^{(k)}||_{2}^{2}) 
			\\&=\frac{\rho}{2}(||\mathbf{Ax}-\mathbf{s}^{(k)}+\frac{1}{\rho}\mathbf{y}^{(k)}||_{2}^{2}-||\frac{1}{\rho}\mathbf{y}^{(k)}||_{2}^{2}) 
			\\&=\frac{\rho}{2}(M||\mathbf{x}-\frac{1}{M}\mathbf{A}^{H}(\mathbf{s}^{(k)}-\frac{1}{\rho}\mathbf{y}^{(k)})||_{2}^{2}-||\frac{1}{\rho}\mathbf{y}^{(k)}||_{2}^{2}
			\\&\quad -\frac{1}{M}||\mathbf{A}^{H}(\mathbf{s}^{(k)}-\frac{1}{\rho}\mathbf{y}^{(k)})||_{2}^{2}+||\mathbf{s}^{(k)}-\frac{1}{\rho}\mathbf{y}^{(k)}||_{2}^{2})\\
			&=\frac{\rho M}{2}||\mathbf{x}-\frac{1}{M}\mathbf{A}^{H}(\mathbf{s}^{(k)}-\frac{1}{\rho}\mathbf{y}^{(k)})||_{2}^{2}+\text{const}, 
		\end{split}
	\end{equation}
	where $\text{const}=-\frac{1}{2\rho}\|\mathbf{y}^{(k)}\|_{2}^{2}-\frac{\rho}{2M}\|\mathbf{A}^H(\mathbf{s}^{(k)}-\frac1\rho\mathbf{y}^{(k)})\|_2^2+\frac{\rho}{2}\|(\mathbf{s}^{(k)}-\frac1\rho\mathbf{y}^{(k)})\|_2^2$. 
\end{proof}


According to Proposition 1, the objective function of \eqref{Ba} can be rewritten as \eqref{opt:powerml}. Then we relax \eqref{Ba}  into an unconstrained optimization problem:
\begin{equation}
	\min_{\mathbf{x}\in \mathbb{C}^N}||\mathbf{x}-\frac{1}{M}\mathbf{A}^{H}(\mathbf{s}^{(k)}-\frac{1}{\rho}\mathbf{y}^{(k)})||_{2}^{2}, \label{opt:relax}
\end{equation}
 The solution to \eqref{opt:relax} is given by
\begin{equation}
	\bar{\mathbf{x}} =  \frac{1}{M}\mathbf{A}^{H}(\mathbf{s}^{(k)}-\frac{1}{\rho}\mathbf{y}^{(k)}),
\end{equation}
where $\bar{\mathbf{x}} $ is the relaxed solution of \eqref{Ba}.

\emph{2)Projection onto the Nonconvex Feasible Set}: Based on the above discussion, now we only need to  project $\bar{\mathbf{x}}$ onto constraints \eqref{BBb} and \eqref{BBc}, which is equivalent to solving the following optimization problem:
\begin{subequations}\label{Va}
	\begin{align}
		\min_{\mathbf{x}}\quad&\left\|\mathbf{x}-\bar{\mathbf{x}}\right\|_2^2\label{Na}\\
		\mbox{s.t.}\quad\
		&\arg(x_{n}-c_{n})<\theta, \,\, n=0,1,\cdots,N-1\label{Vab}\\
		&|x_n|=1, \,\, n=0,1,\cdots,N-1. \label{Vac}
	\end{align}
\end{subequations}
It can be observed that the optimization problem \eqref{Va} is  a unimodular quadratic programming (UQP) problem, which can be solved iteratively using the power method to obtain a local optimum \cite{soltanalian2014}.  The solution $\mathbf{x}^{(k+1)}$ is then updated iteratively with each component $x_n^{(k+1)}$ solved in parallel as follows:
\begin{align}\label{solu:x}
	x_n^{(k+1)} = 
	\begin{cases} 
		e^{j\arg(\bar{x}_n^{(k)})}, & \text{if } |\arg(\bar{x}_n^{(k)} - c_{n})| < \theta, \\
		e^{j(\arg(c_{n}) + \theta)}, & \text{if } \arg(\bar{x}_n^{(k)} - c_{n}) > \theta, \\
		e^{j(\arg(c_{n}) - \theta)}, & \text{if } \arg(\bar{x}_n^{(k)} - c_{n}) < -\theta,
	\end{cases}
\end{align}
where $n=0,1,\cdots,N-1$.


\subsection{Solving Subproblem \eqref{Zb}}

By substituting \eqref{eq:la} into \eqref{Zb}, the subproblem \eqref{Zb} can be transformed into

\begin{subequations}\label{Ma}
	\begin{align}
		\min_{\mathbf{s}\in\mathbb{C}^M}\quad&\mathrm{Re}\left\{\mathbf{y}^{(k)H}\left(\mathbf{A}\mathbf{x}^{(k+1)}-\mathbf{s}\right)\right\}+\frac\rho2\left\Vert\mathbf{A}\mathbf{x}^{(k+1)}-\mathbf{s}\right\Vert_2^2\\
		\mathrm{s.t.}\quad\
		&\frac{\|\mathbf{s}\|_\infty^2}{\frac1M\|\mathbf{s}\|_2^2}\leq\alpha.
	\end{align}
\end{subequations}

Similar to the objective function transformation in Proposition 1, the problem \eqref{Ma} can be reformulated as 
\begin{subequations}\label{Sa}
	\begin{align}
		\min_{\mathbf{s}\in\mathbb{C}^M}\quad&\left\|\mathbf{q}^{(k)}-\mathbf{s}\right\|_2^2\label{Saa}\\
		\mathrm{s.t.}\quad\
		&\frac{\|\mathbf{s}\|_\infty^2}{\frac1M\|\mathbf{s}\|_2^2}\leq\alpha,\label{Aa}
	\end{align}
\end{subequations}
where
\begin{equation}\label{p1}
	\mathbf{q}^{(k)}=\mathbf{A}\mathbf{x}^{(k+1)}+\frac{1}{\rho}\mathbf{y}^{(k)}.
\end{equation}

The problem \eqref{Sa} is still challenging due to the non-convexity of the PAPR constraint. By jointly considering the constraints \eqref{opt0:b} and \eqref{opt0:e}, one can obtain $\|\mathbf{s}\|_{2}^2=M$. Thus,  the  PAPR constraint can be equivalently transformed into the following constraints: 
\begin{equation}
	\|\mathbf{s}\|_{2}^2=M,\,\, \text{and}\,\, |s_m|\leq \sqrt{\alpha}.\label{eq:2eqc}
\end{equation}
Next, we introduce auxiliary variables $\beta$, $\mathbf{v}$ and let $\mathbf{s}=\beta\mathbf{v}$, where $\beta>0$ and $\|\mathbf{v}\|_2^2=1$.
 By substituting  $\mathbf{s}=\beta\mathbf{v}$ into the objective \eqref{Saa} and constraint \eqref{eq:2eqc},  the optimization problem \eqref{Sa} can be reformulated as
\begin{subequations}\label{Da}
	\begin{align}
		\min_{\mathbf{v}\in\mathbb{C}^{M}, {\beta}>0}& \beta^{2}-2\beta\operatorname{Re}\left(\mathbf{v}^{H}\mathbf{q}^{(k)}\right) \label{Db} \\
		\text{s.t.}\quad\
		& |\nu_m|\leq\sqrt{\frac\alpha M},\quad m=0,1,\cdots,M-1 \label{Dc} \\
		&\|\mathbf{v}\|_{2}^{2}=1.\label{Dd}
\end{align}\end{subequations}


Due to the independence of $\beta$ and $\mathbf{v}$ in the optimization problem \eqref{Da}, we can optimize one variable while holding the other fixed, thereby enabling a separable optimization approach.

1) Fix $\beta$: Given $\beta>0$, the problem \eqref{Da} can be expressed as:
\begin{subequations}\label{Fa}
	\begin{align}
		\min_{\mathbf{v}\in\mathbb{C}^M}\quad
		&-\mathrm{Re}\left(\mathbf{v}^H\mathbf{q}^{(k)}\right)\\
		\mathrm{s.~t.~}\quad
		&|\nu_m|\leq\sqrt{\frac\alpha M},\quad m=0,1,\cdots,M-1\\
		&\|\mathbf{v}\|_2^2=1,
\end{align}\end{subequations}
which can be solved using a parallel approach to obtain $\mathbf{v}^{(k+1)}$ \cite{wang2019}:
\begin{equation} \label{Ga} \nu_m^{(k+1)} = \begin{cases} \frac{q_m^{(k)}}{2 \gamma^{(k)}}, & \text{if } \frac{|q_m^{(k)}|}{2 \gamma^{(k)}} < \sqrt{\frac{\alpha}{M}},\\
		\ \sqrt{\frac{\alpha}{M}} e^{j \arg \left( q_m^{(k)} \right)}, & \text{otherwise}, \end{cases} \end{equation}
where $m=0,1,\cdots,M-1$. The value of 
$\gamma^{(k)}$ is determined via  bisection to ensure  $\|\mathbf{v}^{(k+1)}\|_2^2=1$. The details are summarized in TABLE \ref{tbb}.

2) Fix $\mathbf{v}^{(k+1)}$: By substituting the solved value of $\mathbf{v}^{(k+1)}$ into the problem \eqref{Da}, we have
\begin{equation}\min_{\beta\geq0}\beta^{2}-2\beta\operatorname{Re}\left(\mathbf{v}^{(k+1)H}\mathbf{q}^{(k)}\right).\end{equation}
The solution for $\beta$ is given by
\begin{equation}\beta^{(k+1)}=\max\left(\text{Re}\left\{\mathbf{v}^{(k+1)H}\mathbf{q}^{(k)}\right\},0\right).\label{eq:betak1}\end{equation}
Then, based on $\mathbf{s}=\beta \mathbf{v}$, along with \eqref{Ga} and \eqref{eq:betak1}, we obtain  
\begin{equation}\mathbf{s}^{(k+1)}=\mathbf{\beta}^{(k+1)}\mathbf{v}^{(k+1)}.\end{equation}

Thus, the ADMM-PLPOI algorithm to problem \eqref{opt0} is summarized in Algorithm \ref{power1}.
\begin{table}[htbp]
	\centering
	\caption{Bisection Method for Determining Parameter $\gamma^{(k)}$}
	\label{tbb}
	\begin{tabular}{p{0.95\linewidth}}
		\toprule
		\textbf{Initialization for $k = 0$:} \\
		\quad Define search interval $(\gamma_{\mathrm{l}}^{(0)}, \gamma_{\mathrm{r}}^{(0)})$ where \\
		\quad $\bullet$ $\gamma^{(0)} \in (\gamma_{\mathrm{l}}^{(0)}, \gamma_{\mathrm{r}}^{(0)})$ \\
		\quad $\bullet$ $\gamma_{\mathrm{l}}^{(0)} = 0$ \\
		\quad $\bullet$ $\gamma_{\mathrm{r}}^{(0)}$ is sufficiently large \\[2ex]
		
		\textbf{Iterative Process for $k \geq 0$:} \\
		\quad 1. Compute the midpoint: \(\gamma^{(k)} = \frac{\gamma_{\mathrm{l}}^{(k)} + \gamma_{\mathrm{r}}^{(k)}}{2}\). \\
		\quad 2. Update the solution vector \(\mathbf{v}^{(k+1)}\) using equation \eqref{Ga}. \\
		\quad 3. Evaluate \(f(\gamma^{(k)}) = \|\mathbf{v}^{(k+1)}\|_{2}^{2} - 1\): \\
		\quad \qquad $\bullet$ If \(f(\gamma^{(k)}) > 0\): Set \(\gamma_{\mathrm{l}}^{(k+1)} = \gamma^{(k)}\) and \(\gamma_{\mathrm{r}}^{(k+1)} = \gamma_{\mathrm{r}}^{(k)}\). \\
		\quad \qquad $\bullet$ If \(f(\gamma^{(k)}) < 0\): Set \(\gamma_{\mathrm{l}}^{(k+1)} = \gamma_{\mathrm{l}}^{(k)}\) and \(\gamma_{\mathrm{r}}^{(k+1)} = \gamma^{(k)}\). \\[2ex]
		
		\textbf{Termination:} \\
		\quad $\bullet$ Set the final value of \(\gamma^{(k+1)} = \frac{\gamma_{\mathrm{l}}^{(k+1)} + \gamma_{\mathrm{r}}^{(k+1)}}{2}\), if $|f(\gamma^{(k)})|<\epsilon$ . \\
		\bottomrule
	\end{tabular}
\end{table}

\begin{algorithm}[!t]
	\caption{ADMM-PLPOI Algorithm for PD-based Low PAPR OFDM-ISAC Wavefrom Design}
	\label{power1}
	\begin{algorithmic}[0]
		\STATE \textbf{Input:}
		\STATE \hspace{0.5cm} $\bullet$ System parameters: $N$, $M$
		\STATE \hspace{0.5cm} $\bullet$ Control parameters: penalty factor $\rho$, PAPR threshold $\alpha$, PD threshold $\theta \in (0, \pi/2)$ 
		\STATE \hspace{0.5cm} $\bullet$ Signal parameter: pre-defined information signal $\mathbf{c}$
		\STATE \textbf{Initialize:}
		\STATE \hspace{0.5cm} $\bullet$ Variables: $(\mathbf{x}^{(0)}, \mathbf{s}^{(0)}, \mathbf{y}^{(0)})$
		\STATE \textbf{Iterate:} For $K = 1,2,3,\cdots$ until $K > 150$
		\STATE \textbf{step 1:} compute $\mathbf{b}^{(k)} = \frac{1}{M}A^H(\mathbf{s}^{(k)} - \frac{1}{\rho}\mathbf{y}^{(k)})$
		\STATE \hspace{1.2cm} For $i=0,1,2,\cdots,N-1$:
		\STATE \hspace{1.2cm} $x_n^{(k+1)} = \begin{cases} 
			e^{j\arg(b_n^{(k)})} & |\arg(b_n^{(k)}) - c_i| < \theta \\[0.2cm]
			e^{j(c_n+\theta)} & \arg(b_n^{(k)}) - c_n > \theta \\[0.2cm]
			e^{j(c_n-\theta)} & \arg(b_n^{(k)}) - c_n < -\theta
		\end{cases}$
		
		\STATE \textbf{step 2:} compute $\mathbf{q}^{(k)} =\mathbf{ A}\mathbf{x}^{(k+1)} + \frac{1}{\rho}\mathbf{y}^{(k)}$
		\STATE \textbf{step 3:} compute $\mathbf{v}^{(k+1)}$ to find $\gamma^{(k)}$ by binary search
		\STATE \textbf{step 4:} compute $\beta^{(k+1)} = \max\{\text{Re}\{\mathbf{v}^{(k+1)H}\mathbf{q}^{(k)}\}, 0\}$
		\STATE \textbf{step 5:} compute $\mathbf{s}^{(k+1)} = \beta^{(k+1)}\mathbf{v}^{(k+1)}$
		\STATE \textbf{step 6:} compute $\mathbf{y}^{(k+1)} = \mathbf{y}^{(k)} + \rho(\mathbf{A}\mathbf{x}^{(k+1)} - \mathbf{s}^{(k+1)})$
		
		\STATE \textbf{Output:} Optimized signal $\mathbf{x}^{(k+1)}$ and $\mathbf{s}^{(k+1)}$ when stopping criteria is met
		
	\end{algorithmic}
\end{algorithm}

\subsection{A Simple Example}
To demonstrate the effectiveness of our proposed method, we present an illustrative example of the PLPOI waveform design. 

We begin with a conventional communication-only OFDM signal  $\mathbf{s}^c$  whose frequency-domain sequenc $\mathbf{c}$ is generated using random QPSK modulation.  The phases of the frequency-domain sequence $\mathbf{c}$ are displayed in the leftmost column of Table  \ref{arg}.  Initially, this OFDM signal $\mathbf{s}^c$ exhibits a PAPR of 9.7 dB. We then apply our proposed waveform design method with a specified PD threshold
 $\theta=0.6$,   resulting in the PLPOI waveform $\mathbf{s}^x$, with its corresponding frequency-domain sequence $\mathbf{x}$, also shown in Table \ref{arg}. 
 The results confirm that the phase difference  between sequences $\mathbf{c}$ and $\mathbf{x}$ remains confined within the interval  $[-0.6,0.6]$, satisfying the predetermined PD constraint.

 The PAPR reduction performance is visualized in Fig. \ref{fig:papr_example}, which presents the amplitude plots of both $\mathbf{s}^c$ and $\mathbf{s}^x$ alongside their respective PAPR values. The comparison reveals that while the original signal $\mathbf{s}^c$ exhibits large amplitude fluctuations with a PAPR of 9.7 dB, the optimized waveform $\mathbf{s}^x$ demonstrates markedly smaller amplitude variations with a PAPR of only 3.0 dB. This significant reduction in PAPR (a decrease of 6.7 dB) validates the effectiveness of our proposed waveform design method.



\begin{table}[htbp]
	\centering
	\caption{Signal Phase Arguments and PD with $\theta=0.6$}
	\label{arg}
	\begin{tabular}{|c|c|c|}
\hline
$arg(c_n)$ & $arg(x_n)$ & $arg(x_n)-arg(c_n)$ \\ \hline
0.7854 & 0.7895 & 0.0041 \\ \hline
-0.7854 & -0.3726 & 0.4128 \\ \hline
0.7854 & 1.3752 & 0.5898 \\ \hline
0.7854 & 1.2929 & 0.5075 \\ \hline
-0.7854 & -0.3989 & 0.3865 \\ \hline
0.7854 & 0.9513 & 0.1659 \\ \hline
-2.3562 & -2.1435 & 0.2127 \\ \hline
-2.3562 & -2.9562 & -0.6000 \\ \hline
-0.7854 & -1.3686 & -0.5832 \\ \hline
0.7854 & 1.0339 & 0.2485 \\ \hline
\end{tabular}

\end{table}

\begin{figure}[htbp]
	\centering
	\includegraphics[width=0.5\textwidth]{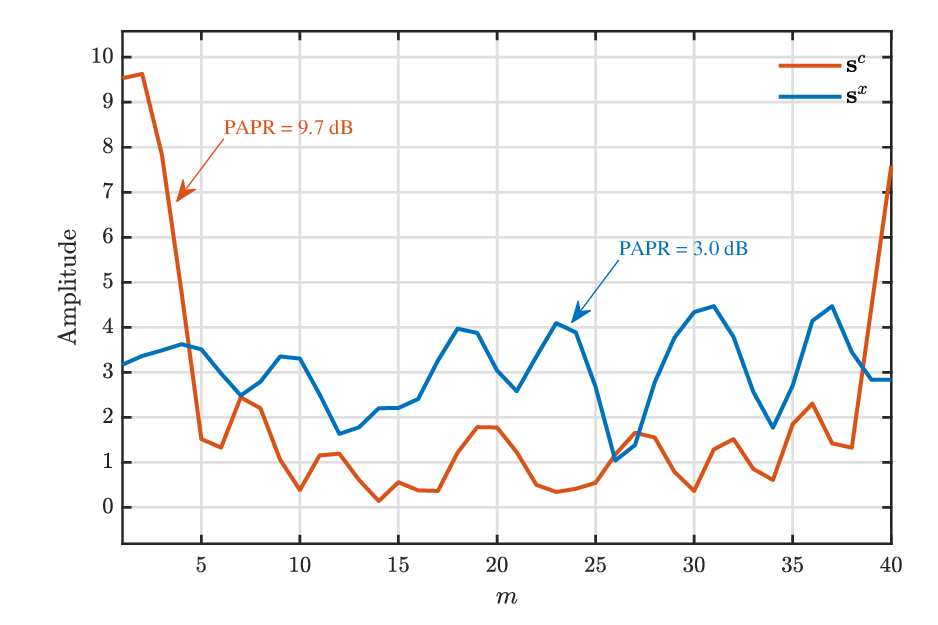}
	\caption{Amplitudes comparison.}
	\label{fig:papr_example}
\end{figure}

%
%

\begin{figure*}[htbp]
	\centering
	\includegraphics[width=1\textwidth]{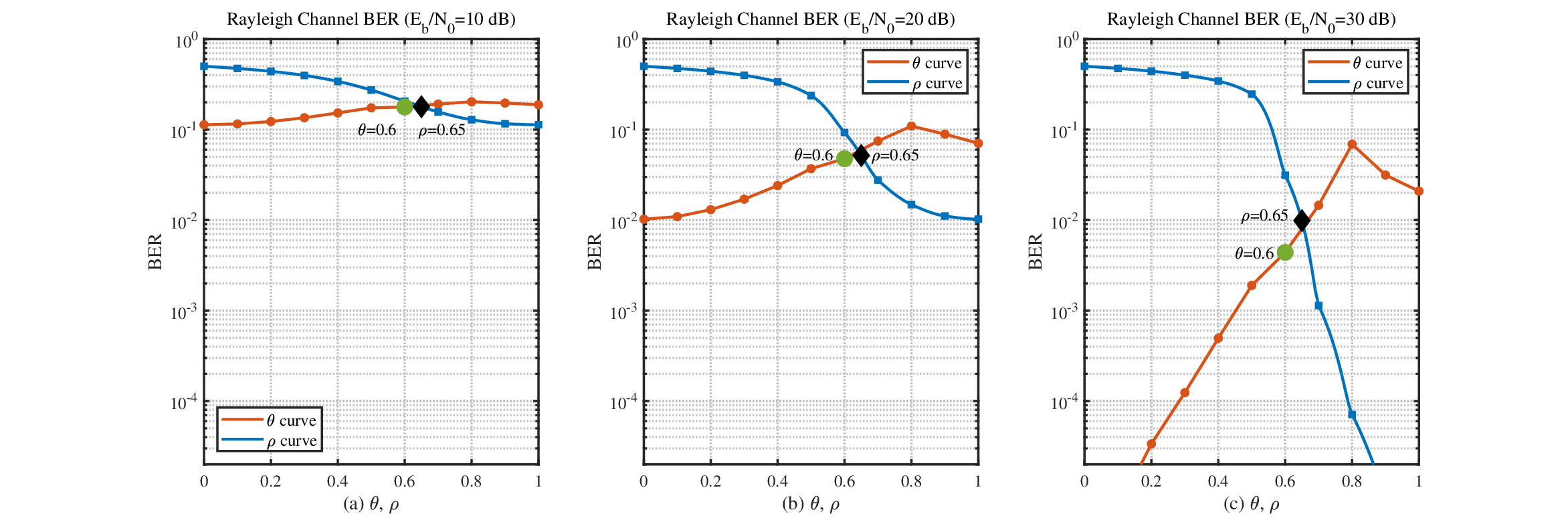}
	\caption{Relationship curves of $\theta$, $\rho$, and BER.}
	\label{output}
\end{figure*}
\subsection{Computational Complexity Analysis }
In Algorithm 1, the computational complexity mainly comes from updating $\mathbf{x}^{(k+1)}$, $ \mathbf{q}^{(k+1)} $, and \( \mathbf{y}^{(k+1) } \), which can be computed by FFT/IFFT. The complexity of FFT/IFFT for length-$N$ \( \mathbf{x} \) is \( \mathcal{O}(N \log_2 N) \) \cite{wang2019}. Consequently, the computational complexity of computing \( \mathbf{x}^{(k+1)} \) in  \eqref{solu:x}, \( \mathbf{q}^{(k+1)} \) in  \eqref{p1}, and \( \mathbf{y}^{(k+1)} \) in \eqref{Zc} are \( \mathcal{O}(M \log_2 M) \), \( \mathcal{O}(N \log_2 N) \), and \( \mathcal{O}(N \log_2 N) \), respectively.  It should be mentioned that the bisection search process for solving the optimal Lagrange multiplier \(\gamma^{(k)} \) typically requires only a small number of iterations, even for large search intervals, and its computational complexity is comparable to or less than a single FFT/IFFT operation \cite{wang2019}. Therefore, it is negligible compared to the FFT/IFFT operations required to compute  $\mathbf{v}^{(k)}$.
	Thus, the total computational complexity of each ADMM-based iteration is roughly $\mathcal{O}(M \log_2 M)$.


\section{Simulation Verification}
In this section, we present simulation results to evaluate the sensing and communications (S\&C) performance of the proposed PLPOI waveforms in OFDM-ISAC systems. The BER is considered as the communication performance metric, while the ambiguity function and radar imaging are used to evaluate sensing performance. Additionally, PAPR is considered for evaluating both S\&C performance. The ISAC system parameters are listed  in TABLE \ref{taba}. 

\begin{table}[!h]
	\centering
	\caption{Simulation Parameters}
	\label{taba}
	\begin{tabular}{llc}
		\toprule
		Symbol & Parameter & Value \\
		\midrule
		\multicolumn{3}{l}{\textbf{OFDM System Parameters}} \\
		$f_c$ & Carrier frequency & 24 GHz \\
		$B$ & Signal bandwidth & 93.1 MHz \\
		$N$ & Number of subcarriers & 1024 \\
		$\Delta f$ & Subcarrier spacing & 90.9 kHz \\
		\midrule
		\multicolumn{3}{l}{\textbf{Timing Parameters}} \\
		$T$ & OFDM symbol duration & 11 $\mu$s \\
		$T_{cp}$ & Cyclic prefix duration & 1.37 $\mu$s \\
		$T_s$ & Total OFDM symbol duration & 12.37 $\mu$s \\
		$G$ & Number of OFDM symbols per generation & 256 \\
		\midrule
		\multicolumn{3}{l}{\textbf{Sensing Parameters}} \\
		$\Delta R$ & Range resolution & 1.61 m \\
		$R_{max}$ & Unambiguous range (limited by $T_{cp}$) & 206.25 m \\
		$\Delta v$ & Velocity resolution & 1.97 m/s \\
		$v_{max}$ & Unambiguous velocity (limited by $\Delta f$) & 113 m/s \\
		\midrule
		\multicolumn{3}{l}{\textbf{Optimization Parameters}} \\
		$\theta$ & Phase difference threshold & 0.6 \\
		$\rho$ & Penalty factor & 10000 \\
		$L$ & Oversampling factor & 4 \\
		$\alpha$ & PAPR threshold & 1.8 dB \\
		\bottomrule
	\end{tabular}
\end{table}

\subsection{PAPR Performance Analysis}
For a comprehensive benchmark comparison, we use the OFDM-ISAC waveform design method presented in~\cite{chen2024} as the primary benchmark due to its representativeness in ISAC waveform design. To ensure a fair evaluation, our proposed scheme is compared with this benchmark under similar conditions.

The benchmark method used in our comparison is derived from the frequency-domain expression (16) in \cite{liu2018}. For the single-antenna case, this optimization problem can be formulated as:

\begin{equation}
	\min_{\mathbf{x}} \rho\|\mathbf{x}-\mathbf{c}\|_2^2 + (1-\rho)\|\mathbf{x}-\mathbf{x_0}\|_2^2
\end{equation}

This formulation is equivalent to equation (62) in Chen et al.~\cite{chen2024} for the single-antenna scenario, where $\rho$ is a weighting parameter controlling the tradeoff between radar and communication performance, $\mathbf{x}$ represents the designed integrated waveform, $\mathbf{c}$ is the communication information signal, and $\mathbf{x_0}$ is the ideal radar waveform with low PAPR characteristics. The closed-form solution to this problem is:
\begin{equation}
	\mathbf{x} = \rho \mathbf{c} + (1-\rho)\mathbf{x_0}
\end{equation}

In our implementation of Chen's method, the corresponding time-domain signal of the selected $\mathbf{x_0}$ has a PAPR of approximately 1.46 dB. As illustrated in Fig.~\ref{output}, under Rayleigh channel conditions with different SNR levels (10, 20, and 30 dB), our method with PD threshold $\theta = 0.6$ and the method in \cite{chen2024} with parameter $\rho = 0.65$ are compared. The following performance comparisons are all based on this set of parameters, i.e., $(\theta, \rho) = (0.6, 0.65)$.

Fig.~\ref{figb} presents the complementary cumulative distribution function (CCDF) comparison of PAPR among different methods: the random OFDM signals, the waveform design method in~\cite{chen2024}, and our proposed PLPOI method introduced in Section~III. It can be observed that the random OFDM signals exhibit a relatively high PAPR value of approximately $12$--$13$ dB in the worst case, while the method in~\cite{chen2024} with $\rho = 0.65$ achieves a PAPR of around $11$--$12$ dB, showing only marginal improvement over the original OFDM. In contrast, our proposed PLPOI method with $\theta = 0.6$ significantly reduces the PAPR to approximately $4$--$5$ dB, demonstrating a substantial 7~dB improvement compared to the benchmark method.

Furthermore, when adjusting our PD threshold to $\theta = 0.7$, our method achieves even better PAPR reduction performance. While both configurations ($\theta = 0.6$ and $\theta = 0.7$) reach similar PAPR values at the $10^{-4}$ probability level, the $\theta = 0.7$ configuration shows more favorable PAPR statistics overall, with its CCDF curve exhibiting lower PAPR values across most of the probability range.

The step-like transitions appearing in the CCDF curves of our proposed method occur because these points rapidly satisfy the strict constraint $\mathbf{s}$ = $\mathbf{A}$$\mathbf{x}$ and reach convergence, limiting the potential for additional PAPR reduction. Despite these characteristics, the overall PAPR reduction performance remains highly effective. This performance improvement can be attributed to the phase difference introduced in our method, which provides greater optimization flexibility across the entire spectrum, while the approach in \cite{chen2024} primarily focuses on balancing the weighted sum of radar similarity and communication multi-user interference (MUI) within a unified optimization framework.

To verify the reliability of this performance improvement, we further investigate the iterative convergence behavior of the PAPR optimization process. As shown in Fig.~\ref{figp}, for $\theta = 0.7$, our proposed method achieves effective PAPR reduction within 50 iterations. This fast and stable convergence characteristic validates the algorithm's practicality in real-world applications.

\begin{figure}[htbp]
	\centering
	\includegraphics[width=0.5\textwidth]{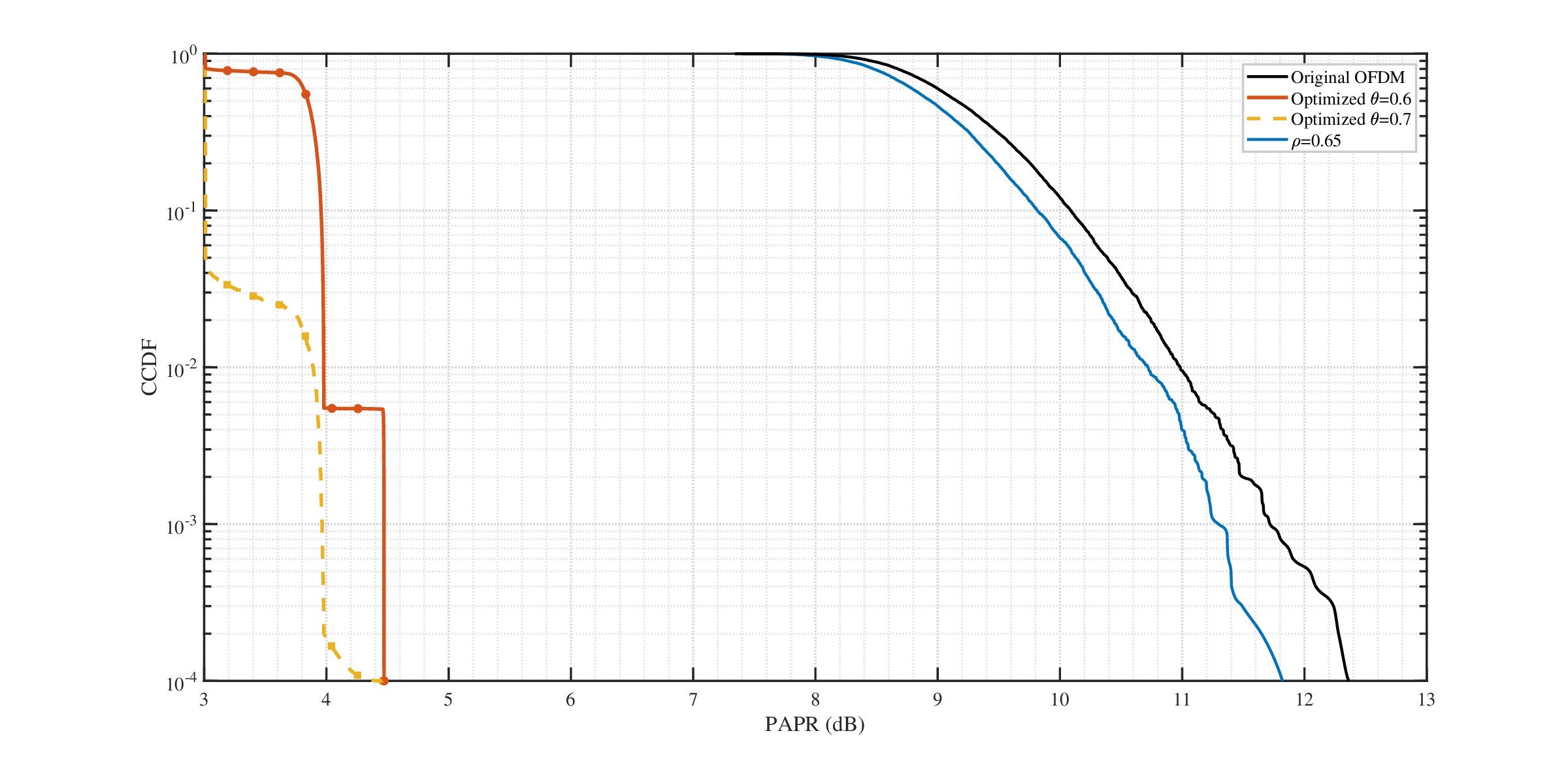}
	\caption{Comparison of PAPR performance.}
	\label{figb}
\end{figure}
\begin{figure}[htbp]
	\centering
	\includegraphics[width=0.5\textwidth]{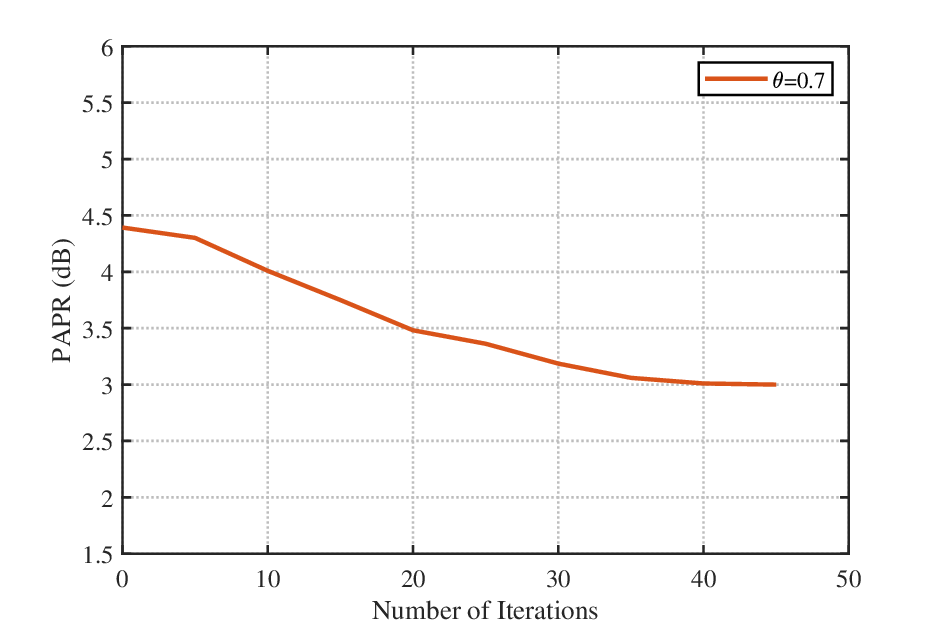}
	\caption{PAPR versus the number of iterations.}
	\label{figp}
\end{figure}





\subsection{Communication Performance Analysis}
\begin{figure}[htbp]
	\centering
	\includegraphics[width=0.5\textwidth]{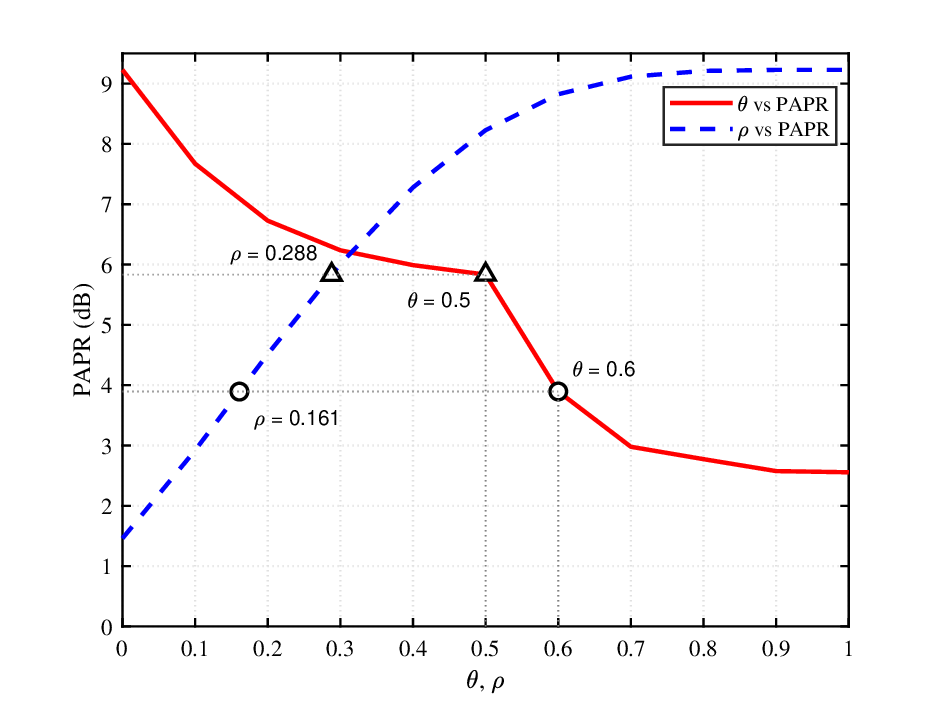}
	\caption{Comparison of PAPR performance with different parameter values.}
	\label{theta10000}
\end{figure}

To enable a fair comparison of communication performance between our proposed method and the benchmark method from~\cite{chen2024}, we identified specific parameter pairs ($\theta$ and $\rho$) that yield equivalent PAPR performance for both approaches. As shown in Fig.~\ref{theta10000}, we established a mapping relationship between these parameters and PAPR values through extensive simulations (10,000 runs), which allowed us to identify intersection points where both methods achieve the same PAPR. Specifically, we selected $\theta = 0.5$ corresponding to $\rho = 0.288$, and $\theta = 0.6$ corresponding to $\rho = 0.161$, as these parameter combinations produce equivalent PAPR performance between the two methods. Based on these matched parameter pairs, we can fairly evaluate the communication performance of both approaches under identical PAPR conditions.

The communication performance is evaluated by BER in both AWGN and Rayleigh fading channels. For QPSK signals in AWGN channel, the theoretical BER expression is given by \cite{cho2002}:
\begin{equation}\label{31}
	P_e = \frac{2(M_{Q}-1)}{M_{Q}\log_2M_{Q}}Q\left(\sqrt{\frac{6E_b}{N_0}\cdot\frac{\log_2M_{Q}}{M_{Q}^2-1}}\right),
\end{equation}
where $M_{Q}=4$ represents the modulation order, and $Q(x)$ is the standard error function defined as:
\begin{equation}\label{32}
	Q(x)=\frac{1}{\sqrt{2\pi}}\int_{x}^{\infty}e^{-t^2/2}dt.
\end{equation}

For Rayleigh fading channel, the theoretical BER is expressed as \cite{simon2000}:
\begin{equation}\label{33}
	P_e=\frac{M_{Q}-1}{M_{Q}\log_2M_{Q}}\left(1-\sqrt{\frac{3\gamma\log_2M_{Q}/(M_{Q}^2-1)}{3\gamma\log_2M_{Q}/(M_{Q}^2-1)+1}}\right),
\end{equation}
where $\gamma=E_b/N_0$ denotes the signal-to-noise ratio.

For the AWGN channel, we examine the BER performance among the theoretical bound, our proposed PD-based approach with $\theta = 0.5$ and $\theta = 0.6$, and the benchmark method in \cite{chen2024} with corresponding parameters $\rho = 0.288$ and $\rho = 0.161$, respectivley. As shown in Fig.~\ref{awg1}, our method with $\theta = 0.5$ achieves substantially better BER performance compared to the method in \cite{chen2024} with $\rho = 0.288$. Specifically, as SNR increases beyond 10 dB, our approach with $\theta = 0.5$ demonstrates significantly superior performance in terms of BER, with the BER curve descending more rapidly compared to the benchmark method with $\rho = 0.288$, eventually achieving nearly an order of magnitude lower error rate at higher SNR values. Similarly, comparing our method with $\theta = 0.6$ to the benchmark with $\rho = 0.161$, we still observe performance advantages, though the gap narrows somewhat. The theoretical curve serves as a reference, showing that all practical implementations exhibit expected performance degradation, with our proposed method maintaining closer alignment to this theoretical bound.

To verify the effectiveness of our approach under more practical channel conditions, we conducted tests under Rayleigh fading channel. As illustrated in Fig.~\ref{ray1}, the performance advantage of our method is maintained in this challenging environment. Specifically, with $\theta = 0.5$, our proposed method consistently outperforms the benchmark method with $\rho = 0.288$ across the entire $E_b/N_0$ range, with the performance gap becoming more pronounced at higher SNR values ($E_b/N_0 > 20$ dB). For example, at $E_b/N_0 = 30$ dB, our method achieves a BER of approximately $10^{-3}$, while the benchmark method remains above $10^{-2}$. Similarly, our method with $\theta = 0.6$ also outperforms the benchmark approach with $\rho = 0.161$ across the evaluated SNR range. These results confirm that the PD-based approach provides reliable communication performance advantages even under challenging fading conditions, while maintaining equivalent PAPR characteristics.

\begin{figure}[!t]
	\centering
	\includegraphics[width=0.5\textwidth]{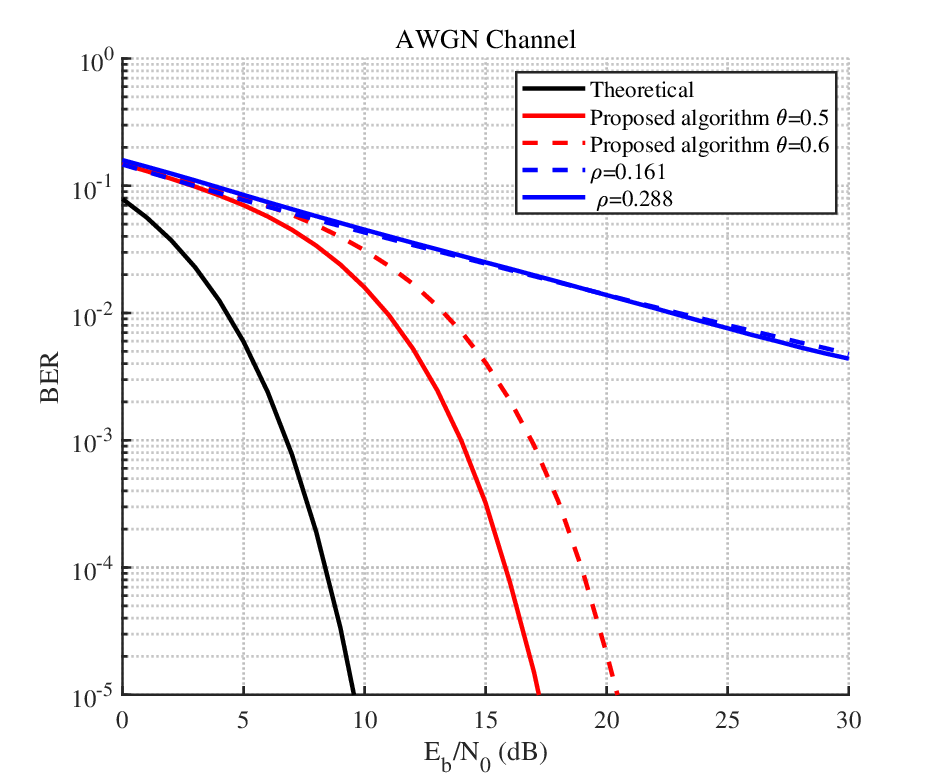}
	\caption{Comparison of Bit Error Rates in AWGN Channel.}
	\label{awg1}
\end{figure}

\begin{figure}[!t]
	\centering
	\includegraphics[width=0.5\textwidth]{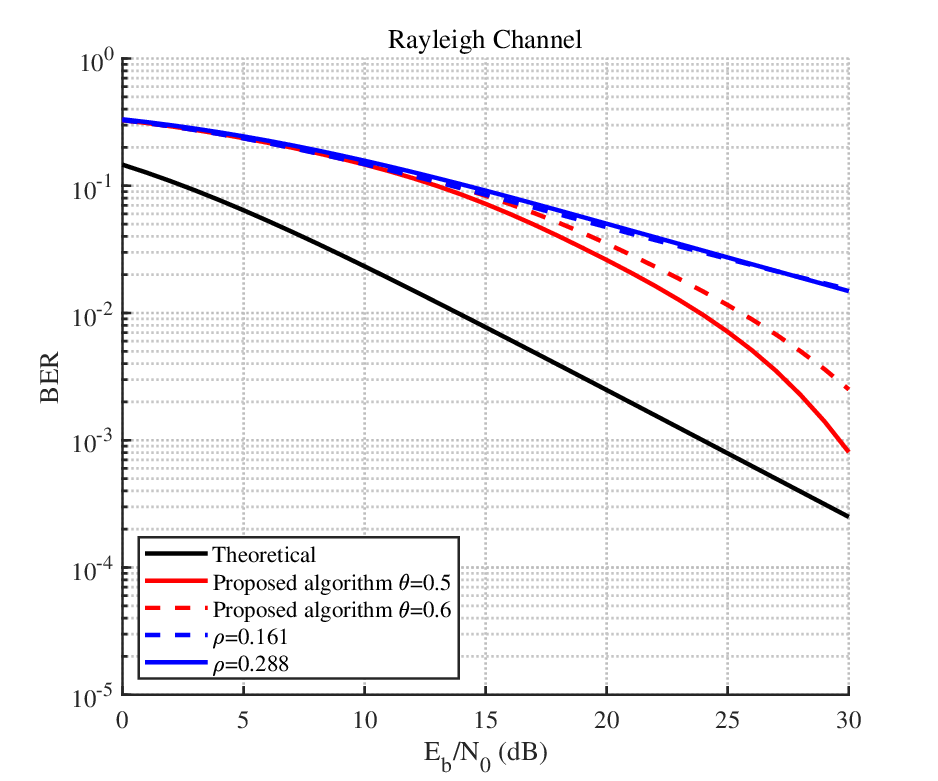}
	\caption{Comparison of Bit Error Rates in Rayleigh Channel.}
	\label{ray1}
\end{figure}

\subsection{Sensing Performance Analysis}
This subsection evaluates the sensing performance of the proposed waveform through analysis of ambiguity function characteristics, multi-target detection capability. We first investigate the ambiguity function characteristics, which is a key metric for evaluating radar system resolution. For the  OFDM signal, the periodic ambiguity function $AF_{\mathbf{s}}(\tau,f)$, $t\in \{0,1,\cdots,N-1\}$, is defined as \cite{ye2022}:

\begin{equation}\label{1}
	AF_{\mathbf{s}}(\tau,f)=\sum_{m=0}^{M-1}s(m)s^*(m+\tau)e^{j2\pi f m/M}.
\end{equation}
where $\tau$ is the delay index and $f$ is the Doppler frequency.

Fig.~\ref{mohu} presents a three-dimensional visualization of the ambiguity function for the  proposed PLPOI waveform without noise, where a sharp peak and rapidly decreasing sidelobes can be clearly observed in both range and velocity dimensions {\footnote{By using the paprmeters in TABLE \ref{taba}, the range and velocity can be calculated by delay and Doppler, respectively.}}. The mainlobe exhibits good concentration, indicating the waveform's potential for high-resolution target detection. The ambiguity function cuts at zero velocity (Doppler) and zero range (delay) provide more detailed insights into the waveform's performance, with further details shown in Fig. \ref{doppwe} and Fig. \ref{juliwei}.

Specifically, Fig.~\ref{doppwe} shows the range profile at zero relative velocity, where the peak sidelobe ratio achieves -330 dB, which benefits from the unimodular property of the frequency domain sequence. This ultra-low sidelobe level facilitates the detection of weak targets in the presence of strong scatterers. Meanwhile, Fig.~\ref{juliwei} presents the Doppler profile at zero range with a peak sidelobe ratio of -28 dB, effectively preventing the masking effect in the velocity dimension. These results demonstrate that the proposed waveform maintains good ambiguity function characteristics essential for radar detection, especially in challenging scenarios with multiple targets or strong clutter.

\begin{figure}[!t]
	\centering
	\includegraphics[width=0.5\textwidth]{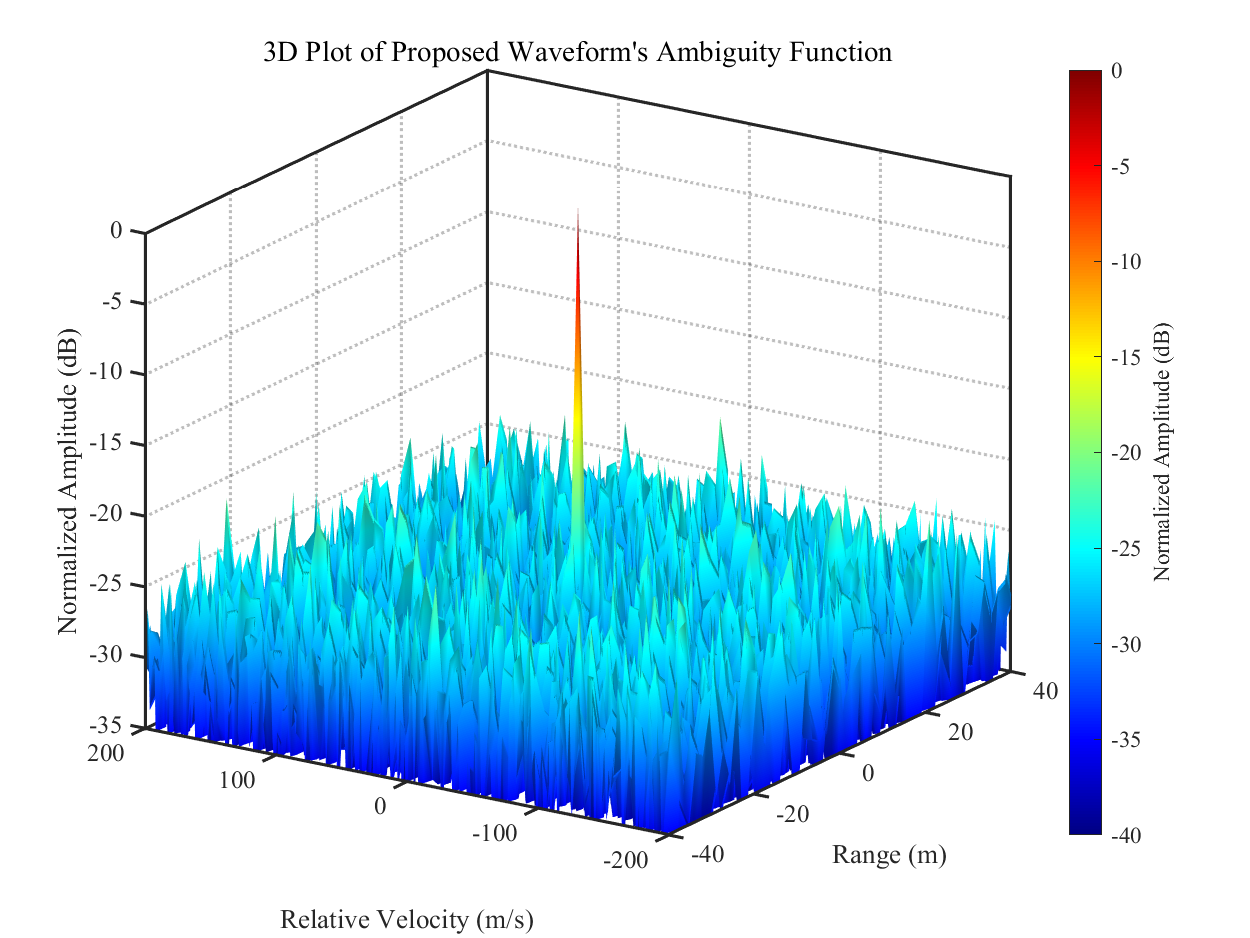}
	\caption{3D Plot of Proposed Waveform's Ambiguity Function.}
	\label{mohu}
\end{figure}

\begin{figure}[!t]
	\centering
	\includegraphics[width=0.5\textwidth]{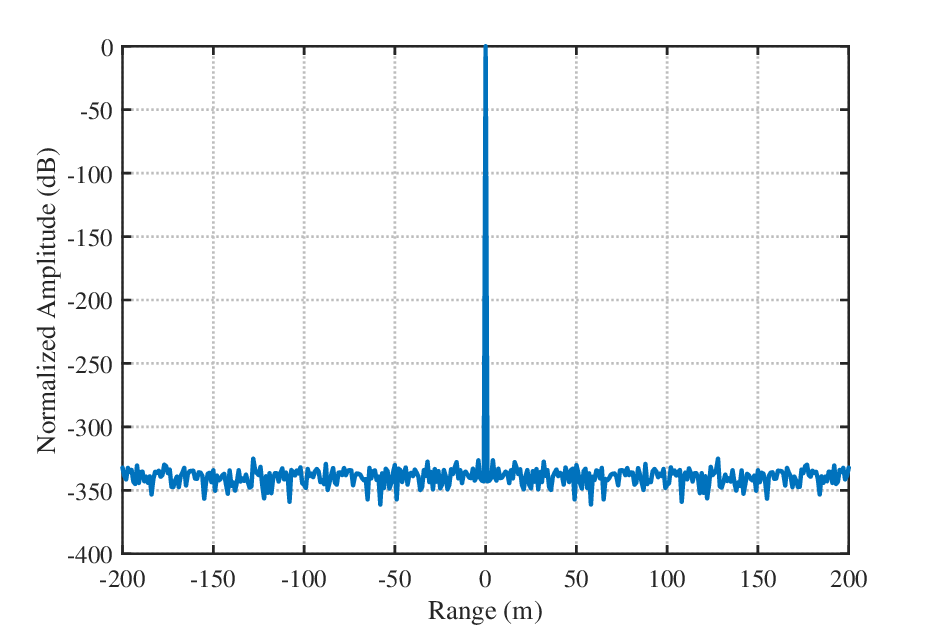}
	\caption{Range profile at 0 m/s relative velocity.}
	\label{doppwe}
\end{figure}

\begin{figure}[!t]
	\centering
	\includegraphics[width=0.5\textwidth]{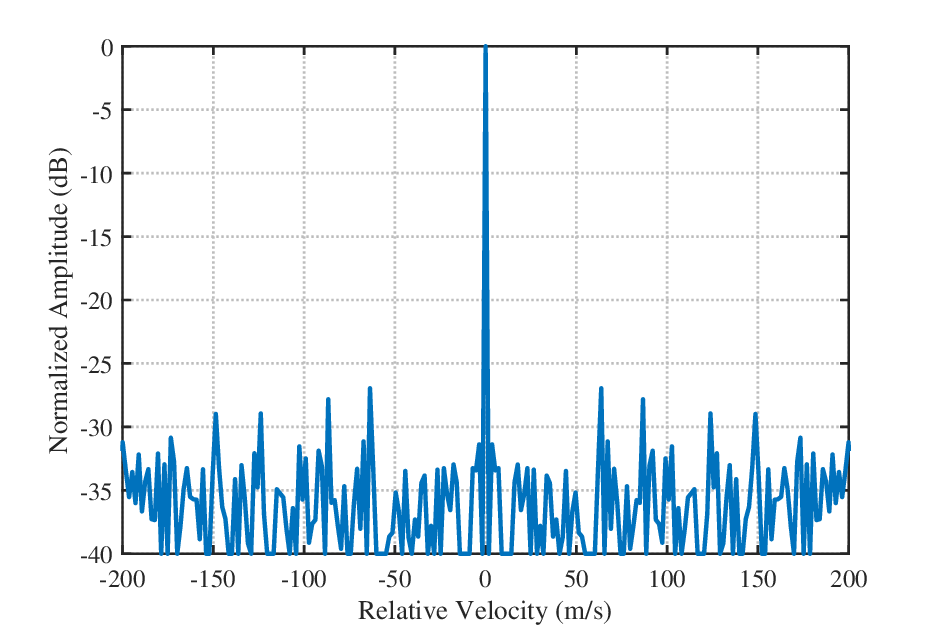}
	\caption{Relative velocity profile at 0 m range.}
	\label{juliwei}
\end{figure}

To further validate the sensing performance in practical scenarios, we conduct multi-target detection simulations. For multiple targets, the received echo signal can be expressed as \cite{cui2017}:     
\begin{equation}\label{2}
	y(m)=b_1{x}(m-\tau_1)e^{j2\pi f_1 m}+b_2x(m-\tau_2)e^{j2\pi f_2 m}+o(m),
\end{equation}     
where $b_1$ and $b_2$ denote the target amplitudes, $o(n)$ represents zero-mean complex Gaussian noise with variance $\sigma^2_v$, and the Doppler frequencies $f_1=2v_1T/\lambda$ and $f_2=2v_2T/\lambda$ correspond to the ranges $R_1$ and $R_2$, respectively.

Fig.~\ref{duomubiao} shows the detection results for two closely spaced targets with parameters specified in TABLE~\ref{tabb}. The first target corresponds to $f_1 = 1.718^{-6}$ with $\tau_1 = 18$, while the second target has $f_2 = 2.577^{-6}$ with $\tau_2 = 21$. The detection result demonstrates several important features of our proposed waveform. First, despite the close spacing between targets (only 4m in range dimension), both targets are clearly distinguishable, indicating excellent range resolution capability. Second, the velocity difference of 5m/s between targets is also well resolved in the Doppler dimension, showing good velocity discrimination ability. Most importantly, unlike traditional waveforms where weaker targets are often masked by the sidelobes of stronger targets, our proposed waveform maintains clear detection of both targets with well-separated peaks and suppressed sidelobes. This is particularly significant in practical radar applications where multiple targets with different radar cross sections often coexist in the same scenario.

\begin{table}[!t]
	\centering
	\caption{Target Parameters for Multi-Target Detection.}
	\label{tabb}
	\begin{tabular}{c|cc}
		\toprule
		\multirow{2}{*}{Target} & \multicolumn{2}{c}{Parameters} \\
		\cmidrule{2-3}
		& Range (m) & Relative Velocity (m/s) \\
		\midrule
		1 & 30 & 10 \\
		2 & 34 & 15 \\
		\bottomrule
	\end{tabular}
\end{table}

\begin{figure}[!t]
	\centering
	\includegraphics[width=0.5\textwidth]{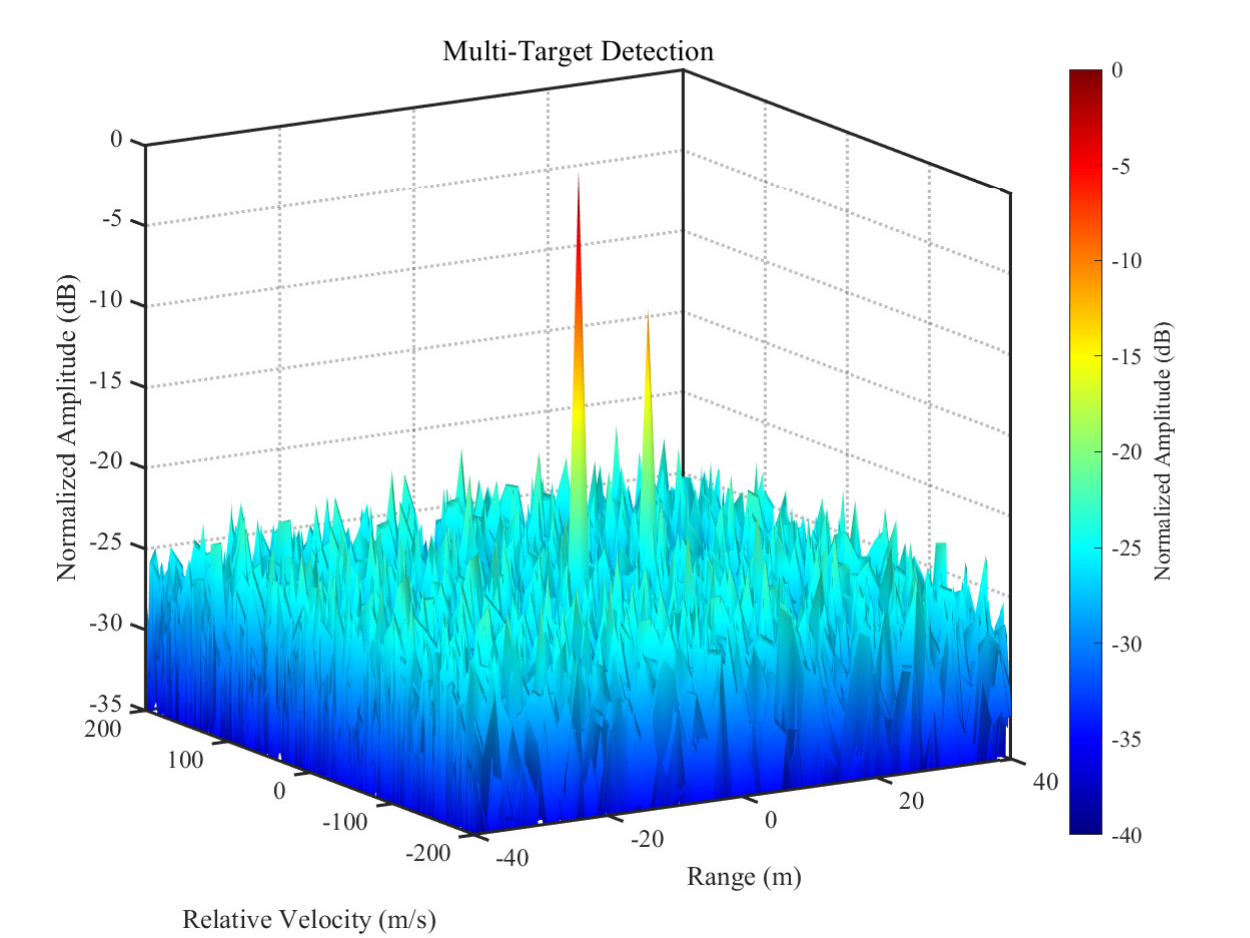}
	\caption{Multiple target scenario.}
	\label{duomubiao}
\end{figure}


These comprehensive results verify that the proposed PD-based waveform design maintains good frequency domain characteristics, achieving excellent sensing performance.



\section{Conclusions}
In this paper, we introduced a novel PD-based OFDM-ISAC waveform structure and designed the corresponding PLPOI waveform. To achieve this, we formulated a non-convex optimization problem that incorporates the time-frequency relationship equation, frequency-domain unimodular constraints, PD conditions, and a time-domain low-PAPR requirement. To efficiently solve this problem, we developed a low-complexity ADMM-PLPOI algorithm, where closed-form solutions are derived for both time-domain and frequency-domain signals. Numerical simulations verify the effectiveness of the proposed PLPOI waveforms in terms of PAPR and BER performance. 




\appendices
\ifCLASSOPTIONcaptionsoff
  \newpage
\fi

\end{document}